\documentclass[11pt,notitlepage]{article}

\newcommand{\cindy}[1]{\textcolor{green!60!black}{(Cindy: #1)}}
\newcommand{\arash}[1]{\textcolor{blue}{(Arash: #1)}}
\newcommand{\bob}[1]{\textcolor{red}{(Bob: #1)}}

\renewcommand{\cindy}[1]{}
\renewcommand{\arash}[1]{}
\renewcommand{\bob}[1]{}

\usepackage{thumbpdf, amssymb, amsmath, amsthm, microtype, array,
graphicx, verbatim, listings, color, fancybox}
\usepackage[pdftex]{hyperref}
\usepackage{mathtools}
\usepackage{amsmath, nccmath}
\usepackage{booktabs}
\usepackage{glossaries}
\usepackage{subcaption}
	
\usepackage{setspace}
\usepackage{acro}
\usepackage{url}
\usepackage[linesnumbered,ruled]{algorithm2e}
\usepackage{subcaption}
\usepackage[utf8]{inputenc}
\usepackage[english]{babel}
\usepackage{caption}
\usepackage{subcaption}
\usepackage{thmtools}
\usepackage{bbm}
\usepackage{pgf,tikz}
\usetikzlibrary{arrows}
\usepackage[margin=1.25in]{geometry}
\usepackage{multicol,multirow}

\usepackage{mathrsfs}
\usepackage{hyperref}
\usepackage{epigraph}
\usepackage{cleveref}

\newcommand\floor[1]{\lfloor#1\rfloor}
\newcommand\ceil[1]{\lceil#1\rceil}

\newcommand{\cov}{\mathrm{cov}}

{
	
\declaretheorem[name=Theorem]{thm}
\declaretheoremstyle[style=claim,qed=$\Diamond$]{claim}
\declaretheoremstyle[style=plain,qed=$\square$]{theorem}

\theoremstyle{plain}

\newtheorem{lemma}[thm]{Lemma}

\newtheorem{proposition}[thm]{Proposition}
\newtheorem{definition}[thm]{Definition}
\newtheorem{claim}{Claim}

\DeclareMathOperator{\spp}{supp}

\DeclareMathOperator{\join}{-\;JOIN}

\DeclareMathOperator{\lp}{LP}

\DeclareMathOperator{\st}{ST}

\DeclareMathOperator{\LP}{LP}
\DeclareMathOperator{\IP}{IP}
\DeclareMathOperator{\dom}{\mathcal{D}}
\DeclareMathOperator{\conv}{conv}

\DeclareMathOperator{\subtour}{{Subtour}}

\DeclareMathOperator{\tsp}{{{TSP}}}

\DeclareMathOperator{\2ec}{{{2EC}}}

\DeclareMathOperator{\LPC}{{{LPC}}}
\DeclareMathOperator{\prun}{{{Pruning}}}
\DeclareMathOperator{\HelperLP}{{{FixNextVarLP}}}

\DeclareMathOperator{\tap}{{{TAP}}}
\DeclareMathOperator{\cut}{{{CUT}}}
\DeclareMathOperator{\permat}{{{PM}}}

\newenvironment{cproof}
{\begin{proof}
 [Proof.]
 \vspace{-1.5\parsep}
}
{ \end{proof}}

\onehalfspacing

\title{Fractional Decomposition Tree Algorithm: A tool for studying the integrality gap of Integer Programs}

\author{\textsc{Robert Carr}\thanks{University of New Mexico  {\tt{bobcarr@unm.edu}}. This material is based upon research supported in part by
the U.S. Office of Naval Research under award number N00014-18-1-2099.}
\and \textsc{Arash Haddadan}\thanks{University of Virginia  {\tt{ahaddada@virginia.edu}}.  This work was mainly done when this author was a graduate student at Carnegie Mellon University.}
\and \textsc{Cynthia A. Phillips}\thanks{Sandia National Laboratories  {\tt{caphill@sandia.gov}}. Sandia National Laboratories is a multi-mission laboratory managed and operated by National Technology and Engineering Solutions of Sandia, LLC., a wholly owned subsidiary of Honeywell International, Inc., for the U.S. Department of Energy’s National Nuclear Security Administration under contract DE-NA0003525.}}

\begin{document}

\maketitle

\begin{abstract}
We present a new algorithm, Fractional Decomposition Tree (FDT) for finding a feasible solution for an integer program (IP) where all variables are binary. FDT runs in polynomial time and is guaranteed to find a feasible integer solution provided the integrality gap is bounded. The algorithm gives a construction for Carr and Vempala's theorem that any feasible solution to the IP's linear-programming relaxation, when scaled by the instance integrality gap, dominates a convex combination of feasible solutions. FDT is also a tool for studying the integrality gap of IP formulations.  We demonstrate that with experiments studying the integrality gap of two problems: optimally augmenting a tree to a 2-edge-connected graph and finding a minimum-cost 2-edge-connected multi-subgraph (2EC). We also give a simplified algorithm, Dom2IP, that more quickly determines if an instance has an unbounded integrality gap. We show that FDT's speed and approximation quality compare well to that of feasibility pump on moderate-sized instances of the vertex cover problem. For a particular set of hard-to-decompose fractional 2EC solutions, FDT always gave a better integer solution than the best previous approximation algorithm (Christofides).
\end{abstract}

\clearpage
\nocite{IPbook}

\section{Introduction}
In this paper we focus on finding feasible solutions to binary Integer Linear Programs (IP). Informally, an integer program is the optimization of a linear objective function subject to linear constraints, where the variables must take integer values. Binary variables represent yes/no decisions. Integer Programming (and more generally Mixed Integer Linear Programming (MILP)) can model many practical optimization problems including scheduling, logistics and resource allocation.

It is  NP-hard even to determine if an IP instance has a feasible solution~\cite{GareyJohnson}. Given its practical importance, however, there are many commercial (e.g. CPLEX, GUROBI, XPRESS) and free (e.g. CBC) solvers that for specific IP instances can often find solutions that are optimal within a given tolerance. Still we have formulated moderate-sized IP instances that no commercial solver can currently solve. Thus there is value in heuristics to find feasible solutions for general IP instances (see e.g.~\cite{HanafiT2017}).
These heuristics, such as the popular Feasibility pump algorithm \cite{fp1,fp2}, are often effective and fast in practice. However, the heuristics can sometimes fail to find a feasible solution. Moreover, these heuristics do not provide any bounds on the quality of the solution they find. 

A major tool for finding feasible solutions for discrete optimization problems expressed as IPs is the {\em linear-programming (LP) relaxation} for the IP formulation. This is a new problem created by relaxing the integrality constraints for an IP instance, allowing the variables to take continuous (rational) values. Linear programs can be solved in polynomial time. The objective value of the linear-programming relaxation provides a bound (lower bound for a minimization problem and upper bound for a maximization problem) on the optimal solution to the IP instance. The solutions can also provide some useful global structure, even though the fractional values might not be directly meaningful. 

{\em LP-based approximation algorithms} use LP relaxations to find provably good approximate feasible solutions to IP problems in polynomial time. At the highest level, they involve solving the LP relaxation, using special structure from the problem to find a feasible solution, and proving that the objective value of the solution is no more than $C$ times worse than the bound from the LP relaxation. The approximation factor $C$ can be a constant or depend on the input parameters of the IP, e.g. $O(\log(n))$ where $n$ is the number of variables in the formulation of the IP (the dimension of the problem).

There is an inherent limit to how small $C$ can be for a given IP. The integrality gap for an IP instance is the ratio of the best integer solution to the best solution of the LP relaxation. Any LP-based approximation cannot have an approximation factor $C$ smaller than the integrality gap because there is no feasible solution with an objective value better than a factor of $C$ worse than the optimal solution of the LP relaxation. 

If the integrality gap for an IP formulation is large, it is sometimes possible to add families of constraints to the formulation to reduce the integrality gap.  These constraints are redundant for the integer problem, but can make some fractional solutions no longer feasible for the LP. These families of constraints (cuts) can have exponential size (number of constraints) as long as we can provide a polynomial-time separation algorithm. A separation algorithm for a family of constraints takes an optimal solution to an LP instance that explicitly enforces only a (potentially empty) subset of the family. It either confirms that all constraints are satisfied or returns a most violated constraint. Thus one can add this constraint and repeat at most a polynomial number of times until all are satisfied.

Reducing the integrality gap of an IP formulation has two advantages.  It can lead to better LP-based approximation algorithm bounds as described above.  It can also help exact solvers run faster or solve instances it could not before. Exact IP solvers are based on intelligent branch-and-bound strategies.  As mentioned above, commercial and open-source MILP solvers can find exact solutions (or near-optimal solutions with a provable bound) to many specific instances of NP-hard combinatorial optimization problems. These solvers use the LP relaxation to get lower bounds (for minimization problems).  The search requires exponential time in the worst case. But this search is practically feasible when the solver can prune large amounts of the search space.  This happens when the lower bound for a region of the search space (subproblem) is worse than the value of a known feasible solution.  This requires a way to find a good heuristic solution and it requires good lower bounds that are as close to the actual optimal value of an IP subproblem as possible.

In this paper, we give a method to find feasible solutions for IPs if the integrality gap is bounded. The method is also a tool for evaluating the integrality gap for a formulation.  Researcher can use it to determine whether they should expend effort to find new classes of cuts.  They can also use it to help guide theory for finding tighter bounds on the integrality gap for classic problems like the traveling salesman problem.

We now describe IPs and our methods more formally. The set of feasible points for a pure IP (henceforth IP) is the set
\begin{equation}
S(A,b)= \{x\in \mathbb{Z}^{n}\;:\; Ax\geq b\}  \label{S},
\end{equation}
where matrix $A$ of rationals has $m$ constraints on $n$ variables and $b \in \mathbb{R}^{m}$.
If we drop the integrality constraints, we have the linear relaxation of set $S(A,b)$,
\begin{equation}
P(A,b) = \{x\in \mathbb{R}^{n}\;:\; Ax\geq b\}. \label{P}
\end{equation}
Let $I=(A,b)$ denote an instance. Then $S(I)$ and $P(I)$ denote $S(A,b)$ and $P(A,b)$, respectively. Given a linear objective function $c$, an IP is $\min \;\{cx:\; x \in S(I)\}$. 

Relaxing the integrality constraints gives the polynomial-time-solvable linear-programming relaxation: $\min \;\{cx:\;x\in P(I) \}$.  The optimal value of this linear program (LP), denoted $z_{\lp}(I,c)$, is a lower bound on the optimal value for the IP, denoted $z_{\IP}(I,c)$. 

Many researchers (see \cite{davids,vazirani}) have developed polynomial time LP-based approximation algorithms that find solutions for special classes of IPs whose cost are provably at most $C\cdot z_{LP}(I,c)$ for some (possibly constant) function $C$. If the analysis uses the LP bound to prove the approximation quality, then $C$ is at least the integrality gap.
\begin{definition}
The integrality gap $g(I)$ for instance $I$ is: $$g(I)= \max_{c\geq 0}\frac{z_{IP}(I,c)}{z_{LP}(I,c)},$$
where $z_{IP}(I,c)$ is the optimal solution to the integer program and $z_{LP}(I,c)$ is the solution to the linear-programming relaxation.
\end{definition}

In general the integrality gap is defined similarly for any objective function, but we wish to be explicit about the class of problems we assume in this paper. For example consider the minimum cost \textsc{2-edge-connected multi-subgraph problem (2EC)}: Given a graph $G=(V,E)$ and $c\in \mathbb{R}^E_{\geq 0}$, 2EC asks for the minimum cost 2-edge-connected subgraph of $G$, with multi-edges allowed. A graph is 2-edge-connected if there are at least two edge-disjoint paths between every pair of vertices.  A linear-programming relaxation for this problem, known as the {\em subtour-elimination} relaxation, is
 \begin{equation}\min \{cx: \sum_{e\in \delta(U)}x_e \geq 2 \mbox{ for } \emptyset \subsetneq U \subsetneq V,\; x\in [0,2]^{E}\}, 
 \label{eq:subtour}
 \end{equation}
where $\delta(U)$ for vertex subset $U$ is the set of edges that cross the cut defined by $U$.  That is, each $e \in \delta(U)$ has one endpoint in $U$ and the other endpoint in $V-U$.
 In this case the instance-specific integrality gap is the integrality gap of the subtour-elimination relaxation for the 2EC on a graph with $n$ vertices. Alexander et al. \cite{alexander2006integrality} showed this instance-specific integrality gap
is at most $\frac{7}{6}$ for instances with $n= 10$ .

The value of $g(I)$ depends on the constraints in~(\ref{S}).  We cannot hope to find solutions for the IP with objective values better than $g(I)\cdot z_{LP}(I,c)$. More generally we can define the integrality gap for a class of instances $\mathcal{I}$ as follows.
\begin{equation}\label{gapproblem}
g(\mathcal{I}) = \max_{c\geq 0 , I\in\mathcal{I}}\frac{z_{IP}(I,c)}{z_{LP}(I,c)}.
\end{equation}
For example, the integrality gap of the subtour-elimination relaxation for the 2EC is at most $\frac{3}{2}$ \cite{wolsey} and at least $\frac{6}{5}$ \cite{alexander2006integrality}. Therefore, we cannot hope to obtain an LP-based $(\frac{6}{5}-\epsilon)$-approximation algorithm for this problem using this LP relaxation to bound the quality of a feasible solution.

Our methods apply theory connecting integrality gaps to sets of feasible solutions. Instances $I$ with $g(I)=1$ have $P(I)=\conv(S(I))$, the convex hull of the lattice of feasible points. In this case, $P(I)$ is an \textit{integral} polyhedron. The spanning tree polytope of graph $G$, $\st(G)$, and the perfect-matching polytope of graph $G$, $\permat(G)$, have this property (\cite{Edmonds2003,edmondsPM}). Thus the linear-programming relaxation for minimum-cost spanning tree has basic feasible solutions (vertices) that are integral solutions, i.e. spanning trees. For such problems there is an algorithm to express vector $x\in P(I)$ (a feasible LP solution) as a convex combination of points in $S(I)$ (feasible IP solutions) in polynomial time \cite{cons-cara}.

\begin{proposition}\label{cara}
	If $g(I)=1$, then for $x\in P(I)$ there exists a positive integer $k$ and $\theta \in [0,1]^k$, where $\sum_{i=1}^{k}\theta_i =1$ and $\tilde{x}^i\in S(I)$ for $i\in [k]$ such that $\sum_{i=1}^{k}\theta_i \tilde{x}^i\leq x$. Moreover, we can find such a convex combination in polynomial time.
\end{proposition}

An equivalent way of describing Proposition \ref{cara} is the following Theorem of Carr and Vempala \cite{Carr2004}. The {\em dominant} of $P(I)$, which we denote by $\dom(P(I))$, is the set of points $x'$ such that there exists a point $x\in P$ with $x'\geq x$ in every component. A polyhedron is of \textit{blocking type} if it is equal to its dominant.

\begin{thm}[Carr, Vempala \cite{Carr2004}] \label{CV2}
	We have $g(I) \leq C$ if and only if for  $x\in P(I)$ there exists $\theta \in [0,1]^k$ where $\sum_{i=1}^{k}\theta_i =1$ and $\tilde{x}^i\in \dom(S(I))$ for $i\in [k]$ such that $\sum_{i=1}^{k}\theta_i \tilde{x}^i\leq Cx$.
\end{thm}
Goemans \cite{goemansblocking} first introduced Theorem \ref{CV2} for blocking-type polyhedra. While there is an exact algorithm for problems with gap $1$, as stated in Proposition \ref{cara}, Theorem~\ref{CV2} is existential, with no construction.
To study integrality gaps, we wish to decompose a suitably scaled linear-programming solution into a convex combination of feasible integer solutions {\bf constructively}.  That is, we ask: assuming reasonable complexity assumptions, given a specific problem $\mathcal{I}$ with  $1<g(\mathcal{I})<\infty$, and $x\in P(I)$ for some $I\in \mathcal{I}$, can we find $\theta \in [0,1]^k$, where $\sum_{i=1}^{k}\theta_i =1$ and $\tilde{x}^i\in S(I)$ for $i\in [k]$ such that $\sum_{i=1}^{k}\theta_i \tilde{x}^i\leq Cx$ in polynomial time? We wish to find the smallest factor $C$ possible.

\subsection{Algorithms and Theory Contributions} 
 
We give a general approximation framework for solving binary IPs.
Consider the set of points described by sets $S(I)$ and $P(I)$ as in (\ref{S}) and (\ref{P}), respectively. Assume in addition that $S(I)\in \{0,1\}^n$ and $P(I)\subseteq [0,1]^n$.
For a vector $x\in \mathbb{R}_{\geq 0}^n$ such that $x\in P(I)$, let the {\em support} of $x$ be $\spp(x)= \{i \in [n]: x_i \neq 0\}$. For an integer $\beta$ let $\{\beta\}^n$ be the vector $y\in \mathbb{R}^n$ with $y_i=\beta$ for $i\in [n]$.

In Section~\ref{sec:binaryfdt} we introduce the \textit{Fractional Decomposition Tree Algorithm} (FDT) which runs in polynomial time algorithm. Given a point $x\in P(I)$ FDT produces a convex combination of feasible points in $S(I)$ that are dominated coordinatewise by a ``factor" $C$ times $x$.
If $C = g(I)$, it would be optimal. However we can only guarantee a factor of $g(I)^{|\spp(x)|}$. FDT iteratively solves linear programs that are about the same size as the description of $P(I)$.

\begin{restatable}{thm}{binaryFDT}
	\label{binaryFDT}
	Assume $1\leq g(I) 	<\infty$. 	
	The Fractional Decomposition Tree (FDT) algorithm, given $x^*\in P(I)$, produces in polynomial time $\lambda\in [0,1]^k$ and $z^1,\ldots,z^k \in S(I)$ such that $k\leq |\spp(x^*)|$, $\sum_{i=1}^{k}\lambda_i z^i\leq \min(Cx^*,\{1\}^{n})$, and $\sum_{i=1}^{k}\lambda_i = 1$. Moreover, $C\leq g(I)^{|\spp(x^*)|}$.
\end{restatable}

In Section~\ref{sec:domTOIP} we describe a subroutine of the FDT, called the DomToIP algorithm, which finds feasible solutions to any IP with finite gap. This can be of independent interest, especially in proving that a model has unbounded gap.
\begin{restatable}{thm}{DomToIP}
	\label{domtoIP}
	Assume $1\leq g(I) < \infty$. The DomToIP algorithm finds $\hat{x}\in S(I)$ in polynomial time.
\end{restatable}


Here is how the FDT algorithm works at a high level for an instance $I$ with LP feasible region $P(I)$: in iteration $i$ the algorithm maintains a convex combination of  vectors in $\mathcal{D}(P(I))$ that have a 0 or 1 value for coordinates indexed $0,\ldots,i-1$. Let $y$ be a vector in the convex combination from iteration $i-1$. We solve a linear-programming problem that gives us $\theta_0,\theta_1\in [0,1]$ and $y^0,y^1\in \mathcal{D}(P(I))$ such that $y\geq \theta_0 y^0 + \theta_1 y^1$, $\theta_0+\theta_1\leq g(I)$, $y^0_i=0$ and $y^1_i=1$. We then replace $y$ in the convex combination with $\frac{\theta_0}{\theta_0+\theta_1}y^0 +\frac{\theta_1}{\theta_0+\theta_1}y^1$.  Repeating this for every vector in the convex combination from the previous iteration yields a new convex combination of points. It is ``more'' integral because now all vectors in the convex combination are integral in their first $i+1$ elements. If in any iteration there are too many points in the convex combination,  we solve a linear-programming problem that ``prunes'' the convex combination. At the end we have a convex combination of integer solutions $\mathcal{D}(P(I))$. For each such solution $z$ we invoke the DomToIP algorithm
to find $z'\in S(I)$ where $z'\leq z$.

One can extend the FDT algorithm for binary IPs into covering\footnote{A covering IP has nonnegative constraint matrix, objective coefficients and right-hand side ($A,b,c$).} $\{0,1,2\}$ IPs by losing a factor $2^{|\spp(x)|}$ on top of the loss for FDT. To eradicate this extra factor, we must treat the coordinate $i$ with $x_i=1$ differently. In Section~\ref{sec:2EC} we focus on the 2-edge-connected multi-subgraph graph problem (2EC) defined above.  It's subtour-elimination LP relaxation is given in (\ref{eq:subtour}).
For input graph $G$, let $\subtour(G)$ denote the feasible region of this LP. Let $\2ec(G)$ be the convex hull of incidence vectors\footnote{The incidence vector $x$ of a 2-edge-connected multi-subgraph has an element for each edge $e$, with $x_e \in \{0,1,2\}$ indicating the number of times edge $e$ appears in the solution.} of 2-edge-connected multi-subgraphs of graph $G$. Following the definition in (\ref{gapproblem}) have
\begin{equation}
g(\2ec) = \max_{c\geq 0 , G}\frac{\min_{x\in \2ec(G)} cx}{\min_{x\in \subtour(G)} cx}.
\end{equation}

\begin{restatable}{thm}{FDTEC}
	\label{FDT2EC}
	Let $G=(V,E)$ and $x$ be an extreme point of  $\subtour(G)$. The FDT algorithm for 2EC produces $\lambda\in [0,1]^k$ and 2-edge-connected multi-subgraphs $F_1,\ldots,F_k$ such that $k\leq 2|V|-1$, $\sum_{i=1}^{k}\lambda_i \chi^{F_i}\leq \min(Cx,\{2\}^n)$, and $\sum_{i=1}^{k}\lambda_i = 1$. Moreover, $C\leq g(\2ec)^{|E_x|}$.
\end{restatable}

\subsection{Experiments.} Although the bounds guaranteed in both Theorems \ref{binaryFDT} and \ref{FDT2EC} are large, in Section~\ref{sec:experiment} we show that in practice, the algorithm can work well for network design problems like those described above. We also show how one might use FDT to investigate the integrality gap for such well-studied problems.

\subsubsection{Minimum vertex cover problem}

In the \textsc{minimum vertex cover problem (VC)} we are given a graph $G=(V,E)$ and $c\in \mathbb{R}^E_{\geq 0}$. A subset of $U$ of $V$ is a \textit{vertex cover} if for all $e\in E$ at least one endpoint of $e$ is in $U$. The goal in VC is to find the minimum-cost vertex cover. The linear-programming relaxation for VC is
\begin{equation}
\min \{cx \; : \; x_u + x_v \geq	 1 \text{ for all } e=uv \in E, \; x\in [0,1]^{V}\}.
\end{equation}
The integrality gap of this formulation is exactly 2 \cite{davids}. Austrin, Khot and Safra~\cite{UGhardVC} show that it is UG-hard to approximate VC within any factor strictly better than 2. We compare  FDT and the feasbility pump heuristic \cite{fp1} on the small instances of the PACE\footnote{ Parameterized Algorithms and Computational Experiments: \url{https://pacechallenge.org/2019/}} 2019 challenge test cases \cite{PACE}. FDT was 3-5x slower, but still ran in seconds. It always gave a better vertex cover. For some instances the FDT solution's relative gap with respect to the optimal was half that of feasibility pump.
\subsubsection{Tree augmentation problem}
In the \textsc{Tree Augmentation Problem (TAP)} we are given a  graph $G=(V,E)$, a spanning tree $T$ of $G$ and a cost vector $c\in \mathbb{R}^{E\setminus T}_{\geq 0}$. A subset $F$ of $E\setminus T$ is called a \textit{feasible augmentation} if $(V,T\cup F)$ is a 2-edge-connected graph. In TAP we seek the minimum cost feasible augmentation. The natural linear-programming relaxation for TAP is 
\begin{equation}\label{eq:cutLP}
\min \{cx\; : \; \sum_{\ell \in \cov(e)} x_{\ell} \geq 1 \text{ for } e\in T, \; x\in [0,1]^{E\setminus T}\}.
\end{equation}
where $\cov(e)$ for $e \in T$ is the set of edges $\ell \in E\setminus T$ such that $e$ is in the unique cycle of $T\cup \{\ell\}$. Another way to think of $\cov(e)$ is that if we remove edge $e$ from tree $T$, this partitions the vertices into two connected sets $U$ and $V-U$.  Then $\cov(e) = \delta(U)$ is the set of edges $\ell \neq e$ that cross the cut between $U$ and $V-U$, and whose addition would reconnect the tree. We call the LP above the cut LP. The integrality gap of the cut LP is known to be between $\frac{3}{2}$ \cite{32gaptap} and $2$ \cite{FJ81}. We create random fractional extreme points of the cut LP and apply FDT to find integral solutions. For our instances, the ratio of the value of the feasible solution to the LP lower bound is always below $\frac{3}{2}$. This provides evidence that the integrality gap for such instances may be less than $\frac{3}{2}$.

\subsubsection{2-edge-connected multi-subgraph problem}
\label{sec:2EC-intro}
Known polyhedral structure makes it easier to study integrality gaps for such problems. We use the idea of fundamental extreme point \cite{carrravi,boydcarr,Carr2004} to create the ``hardest'' LP solutions to decompose.

There are fairly good bounds for the integrality gap for the Traveling Salesman Problem (TSP, see \cite{tspbook}) or 2EC.
Benoit and Boyd~\cite{TSPcompute} used a quadratic program to show the integrality gap of the subtour-elimination relaxation for the TSP, $g(\tsp)$, is at most $\frac{20}{17}$ for graphs with at most 10 vertices. Alexander et al. \cite{alexander2006integrality} used the same ideas to provide an upper bound of $\frac{7}{6}$ for $g(\2ec)$ on graphs with at most 10 vertices. 

Consider a graph $G=(V,E)$. A \textit{Carr-Vempala point} $x\in \mathbb{R}^E$ is a fractional point in $\subtour(G)$ where the edges $e$ with $0<x_e<1$ form a single cycle in $G$ and the vertices on the cycle are connected via vertex-disjoint paths of edges $e$ with $x_e =1$ (see Figure \ref{fig:CVpoint}).

Carr and Vempala \cite{Carr2004} showed that $g(\2ec)$ is achieved for instances where the optimal solution to $\min_{x\in \subtour(G)}cx$ is a Carr-Vempala point. Multiple groups of researchers have conjectured that $g(\2ec)\leq \frac{6}{5}$ (see \cite{alexander2006integrality,boydlegault,hn19}, but the only known upper bound on $g(\2ec)$ stands at $\frac{3}{2}$ \cite{wolsey}. We show that the integrality gap is at most $\frac{6}{5}$ for Carr-Vempala points with at most 12 vertices on the cycle formed by the fractional edges, thereby confirming the conjecture for these instances.
\cindy{This sounds like a contribution of the paper.} \arash{Yes, particlarly since a similar research question have been studied by \cite{alexander2006integrality} and \cite{TSPcompute}}
Note that the number of vertices in these instances can be arbitrarily high since the paths of edges with $x$-value 1 can be arbitrarily long.

\cindy{We probably need an explicit list of contributions.  It could be another subsection of the introduction.  This will help the reviewer and eventually, if published, the reader.}
\subsection{Contribution summary}
In summary, this papers's contributions are:
\begin{itemize}
\item We give an algorithm to express any feasible point for the LP relaxation of a binary IP
as a convex combination of feasible solutions, provided the IP integrality gap is bounded.  This is the Fractional Decomposition Tree (FDT) algorithm.
\item We give a simple  algorithm to give a feasible solution for any binary IP with bounded integrality gap. If this algorithm fails, then the integrality gap is infinite.
\item We experimentally show that FDT can give good approximate solutions to small instances of vertex-cover and tree augmentation problems.
\item We show that the integrality gap for 2EC (minimum 2-edge-connected sub-multigraph) is at most $\frac{6}{5}$ for Carr-Vempala points with at most $12$ vertices on the cycle formed by the fractional edges, matching the lower bound for integrality gap for 2EC for these instances.
\item We demonstrate (with tree augmentation and 2EC) how FDT exeriments can support theoretical work on integrality gaps for specific problems.
\end{itemize}

\section{Finding a Feasible Solution}\label{sec:domTOIP}
In this section we give the algorithm for DomToIP and prove its performance (Theorem~\ref{domtoIP}).

Consider an IP instance $I=(A,b)$. Define sets $S(I)$ and $P(I)$ as in (\ref{S}) and (\ref{P}), respectively. Assume $S(I)\subseteq \{0,1\}^n$ and $P(I)\subseteq [0,1]^n$. For simplicity in the notation we assume an instance $I$ and denote $P(I),S(I),$ and $g(I)$ by $P$, $S$, and $g$ for this section and the next section. Also, for both sections let $x^*$ be the optimal solution to the LP formulation and assume $t=|\spp(x^*)|$. Reorder variables as necessary to put all the nonzero variables into the first $t$ indices. So we can assume $x^*_i = 0$ for $i=t+1,\ldots,n$.

In this section we prove Theorem \ref{domtoIP}. In fact, we prove a stronger result. 
\begin{lemma}\label{domlemma}
	Given an integral vector in the dominant of an LP relaxation ($\tilde{x}\in \dom(P)$ and $\tilde{x}\in \{0,1\}^n$) for an IP instance that has integrality gap $g < \infty$, there is an algorithm (the DomToIP algorithm) that finds $\bar{x}\in S$ in polynomial time such that $\bar{x}\leq \tilde{x}$.\end{lemma}
Lemma \ref{domlemma} implies Theorem \ref{domtoIP}, since it is easy to obtain an integer point in $\dom(P)$: numerically rounding up $x^*$ (or any fractional point in $P$) gives us a point in $\dom(P)$.
Hence, we can assume in the proof below that $\tilde{x}_i= 0$ for $i=t+1,\ldots,n$. 
\subsection{Proof of Lemma \ref{domlemma}: The DomToIP Algorithm}

We introduce an algorithm that builds a solution from the input integral point $\tilde{x}$. It finalizes a binary value for each variable iteratively, starting from the first coordinate and ending at the $t$-th coordinate.
In iteration $\ell\in \{0,\ldots,t-1\}$, if $\tilde{x}_{\ell+1}=1$, the algorithm reduces the $\ell + 1$st coordinate of the solution to $0$ if there is an LP-feasible point $x'$ that respects the decisions made so far, and has $x'_{\ell}=0$. Intuitively, we can fix this coordinate to $0$ in the partial solution and still set the remaining coordinates to find the solution promised in Lemma~\ref{domlemma}.
\arash{A caution here: the phrase if possible a bit vague here. We round down to zero, if there is feasible solution to LP (so a point in $P(I)$, such that the first $\ell$ coordinates of it are all equal to the point $x^{(\ell)}$ that we have in iteration $\ell$ and its $\ell+1$-th coordinate is 0}.
\cindy{Yes.  I was just trying to give intuition. In particular, I was trying to distinguish Dom2IP that starts with an integral vector, where we push $1$s down, from FDT, where we start with a fractional vector and progressively set components to binary values. Do you think the above version is OK?}
More specifically, in iteration $\ell\in \{0,\ldots,t-1\}$ we produce $x^{(\ell)}\in \dom(P)$ such that $x^{(\ell)}_i\in \{0,1\}$ for $i=1,\ldots,\ell$. These $\ell$ components will not change for the remainder of the algorithm.

We can set $x^{(0)}=\tilde{x}$. We now show how to find $x^{(\ell+1)}$ from $x^{(\ell)}$. Consider the following linear program. The variables of this linear program are the $z\in \mathbb{R}^n$.
\begin{align}
\HelperLP(x^{(\ell)})\quad\quad& \min\quad \;z_{\ell+1}\\
&\;\text{s.t.} \quad \;\;Az\geq b \\
&\;{\color{white}{\text{s.t.}} }\quad \;\; \; z_j = x^{(\ell)}_j \quad \; j =1,\ldots, \ell\\
&\;{\color{white}{\text{s.t.}} }\quad \; \;\; z_j \leq x^{(\ell)}_j \quad \; j = \ell+1,\ldots,n\\
&\;{\color{white}{\text{s.t.}} }\quad \; \;\; z\;\geq 0
\end{align}

If the optimal value to $\HelperLP(x^{(\ell)})$ is 0, then let $x^{(\ell+1)}_{\ell+1} = 0$. Otherwise if the optimal value is strictly positive let $x^{(\ell+1)}_{\ell+1} = 1$. Let $x^{(\ell+1)}_j = x^{(\ell)}_j$ for $j\in [n]\setminus \{\ell+1\}$. The DomToIP algorithm initializes with $x^{(0)}=\tilde{x}$ and iteratively calls this procedure in order to obtain $x^{(t)}$ (See Algorithm \ref{domtoIPalg}).
\cindy{I rearranged the above paragraph a bit without changing the content.  Please see if you think it's OK.}

\vspace*{10pt}
\begin{algorithm}[h]
	\KwIn{$\tilde{x}\in \dom(P)$, $\tilde{x}\in \{0,1\}^n$ }
	\KwOut{$x^{(t)} \in S$, $x^{(t)}\leq \tilde{x}$}
	$x^{(0)}\leftarrow \tilde{x}$\\
	\For{$\ell = 0$ \textbf{to} $t-1$}{
		$x^{(\ell+1)} \leftarrow x^{(\ell)}$\\
		$\eta \leftarrow$ optimal value of $ \HelperLP(x^{(\ell)})$\\
		\eIf{$\eta = 0$}{
			$x^{(\ell+1)}_{\ell+1} \leftarrow 0$\
		}{
			$x^{(\ell+1)}_{\ell+1} \leftarrow 1$
		}
	}
	\caption{The DomToIP algorithm}
	\label{domtoIPalg}
\end{algorithm}
\vspace*{10pt}

We prove that indeed $x^{(t)}\in S$. First, we need to show that in any iteration $\ell=  0,\ldots,t-1$ of DomToIP that $\HelperLP(x^{(\ell)})$ is feasible. We show something stronger. For $\ell=0,\ldots,t-1$ let
\begin{align*}
\LP^{(\ell)}&= \{z\in P\; : \; z\leq x^{(\ell)} \mbox{ and } z_j=x_j^{(\ell)} \mbox{ for } j\in [\ell]\}, \text{ and}\\
\IP^{(\ell)}&= \{z\in \LP^{(\ell)}\; : \; z\in \{0,1\}^n\}.
\end{align*}
If $\LP^{(\ell)}$ is a non-empty set then $\HelperLP(x^{(\ell)})$ is feasible. We show by induction on $\ell$ that $\LP^{(\ell)}$ and $\IP^{(\ell)}$ are not empty sets for $\ell=0,\ldots,t-1$. First, $\LP^{(0)}$ is feasible since by definition $x^{(0)}\in \dom(P)$, meaning there exists $z\in P$ such that $z\leq x^{(0)}$. By Theorem \ref{CV2}, there exists $\tilde{z}^i\in S$ and $\theta_i\geq 0$ for $i\in [k]$ such that $\sum_{i=1}^{k} \theta_i = 1$ and $\sum_{i=1}^{k}\theta_i \tilde{z}^i \leq gz$. (This assumes a finite integrality gap $g$).
Hence, $\sum_{i=1}^{k}\theta_i \tilde{z}^i \leq gz\leq gx^{(0)}$. So if $x^{(0)}_j=0$, then $ \sum_{i=1}^{k}\theta_i \tilde{z}_j^i =0$, which implies that $\tilde{z}^i_j=0$ for all $i\in [k]$ and $j\in [n]$ where $x^{(0)}_j=0$. Hence, $\tilde{z}^i\leq x^{(0)}$ for $i\in [k]$. Therefore $\tilde{z}^i\in \IP^{(0)}$ for $i\in [k]$, which implies $\IP^{(0)}\neq \emptyset$.

Now assume $\IP^{(\ell)}$ is non-empty for some $\ell \in [t-2]$. Since $\IP^{(\ell)}\subseteq\LP^{(\ell)}$ we have $\LP^{(\ell)}\neq \emptyset$ and hence the $\HelperLP(x^{(\ell)})$ has an optimal solution $z^*$.

We consider two cases. In the first case, we have $z^*_{\ell+1}=0$. In this case we set $x^{(\ell+1)}_{\ell+1}=0$. Since $z^*\leq x^{(\ell+1)}$, we have $z^*\in \LP^{(\ell+1)}$. Also, $z^*\in P$. By Theorem \ref{CV2} there exists $\tilde{z}^i\in S$ and $\theta_i\geq 0$ for $i\in [k]$ such that $\sum_{i=1}^{k} \theta_i = 1$ and  $\sum_{i=1}^{k}\theta_i \tilde{z}^i \leq gz^*$. We have $\sum_{i=1}^{k}\theta_i \tilde{z}^i \leq gz^*\leq gx^{(\ell+1)}$.
So for $j\in [n]$ where $x^{(\ell+1)}_j=0$, we have $z^i_j=0$ for $i\in [k]$. This implies $\tilde{z}^i\leq x^{(\ell+1)}$ for $i=1,\ldots,k$. Hence, there exists $z\in S$ such that $z\leq x^{(\ell+1)}$. We claim that $z\in \IP^{(\ell+1)}$. If $z\notin \IP^{(\ell+1)}$ we must have $1\leq j \leq \ell$ such that $z_j < x^{(\ell+1)}_{j}$, and thus $z_j = 0$ and $x^{(\ell+1)}_j=1$. Without loss of generality assume $j$ is the minimum index (between $1$ and $\ell$) satisfying $z_j < x^{(\ell+1)}_{j}$. Consider iteration $j$ of the DomToIP algorithm. We have $z\leq x^{(\ell+1)}\leq x^{(j)}$.
We also have $x^{(j)}_j=1$ which implies when we solved $\HelperLP(x^{(j-1)})$ the optimal value was strictly larger than zero. However, $z$ is a feasible solution to $\HelperLP(x^{(j-1)})$ and gives an objective value of 0. This is a contradiction, so $z\in \IP^{(\ell+1)}$.

Now for the second case, assume $z^*_{\ell+1} > 0$. We set $x^{(\ell+1)}_{\ell+1}=1$. For each point $z\in \LP^{(\ell)}$ we have $z_{\ell+1} >0$, so for each $z\in \IP^{(\ell)}$ we have $z_{\ell+1}>0$, i.e. $z_{\ell+1}=1$. This means that $z\in \IP^{(\ell+1)}$, and $\IP^{(\ell+1)} \neq \emptyset$.

Now consider $x^{(t)}$. Let $z$ be the optimal solution to $\LP^{(t-1)}$. If $x^{(t)}_t = 0$, we have $x^{(t)} = z$, which implies that $x^{(t)}\in P$, and since $x^{(t)}\in \{0,1\}^n$ we have $x^{(t)}\in S$. If $x^{(t)}_t =1$, it must be the case that $z_t > 0$. By the argument above there is a point $z'\in \IP^{(t-1)}$. We show that $x^{(t)} = z'$. For $j\in [t-1]$ we have $z'_j= x_j^{(t-1)}=x_j^{(t)}$. We just need to show that $z'_t = 1$. Assume $z'_t	 = 0$ for contradiction. Then $z'\in \LP^{(t-1)}$ has objective value $0$ for $\HelperLP(x^{(t-1)})$. This is a contradiction to $z$ being the optimal solution to $\LP^{(t-1)}$. This concludes the proof of Lemma \ref{domlemma}.

\section{FDT on Binary IPs}
\label{sec:binaryfdt}
Recall that $x^*$ was the optimal solution to minimizing a cost function $cx$ over set $P$, which provides a lower bound on $\min_{(x,y)\in S(I)} cx$.  In this section, we prove Theorem \ref{binaryFDT} by describing the Fractional Decomposition Tree (FDT) algorithm. We also remark that if $g(I)=1$, then the algorithm gives an exact decomposition of any feasible solution.

The FDT algorithm grows a tree similar to the classic branch-and-bound search tree for integer programs. Each node represents a partially integral vector $\bar{x}$ in $\dom(P)$ together with a multiplier $\bar{\lambda}$. The solutions contained in the nodes of the tree become progressively more integral at each level. In each level of the tree, the algorithm maintain a conic combination of points with the properties mentioned above. Leaves of the FDT tree contain integral solutions
that dominate a point in $P$. In Lemma~\ref{domlemma} we saw how to turn these into points in $S$. 

\paragraph{Branching on a node.}
We begin with the following lemmas that show how the FDT algorithm branches on a variable.
\begin{lemma}\label{LPClemma}
	Given $x'\in \dom(P)$ and $\ell\in [n]$ where $x'_{\ell}<1$, in polynomial time we can find vectors $\hat{x}^0,\hat{x}^1$ and scalars $\gamma_0,\gamma_1 \in [0,1]$ such that: (i) $\gamma_0 + \gamma_1  \geq 1/g$, (ii) $\hat{x}^0$ and $\hat{x}^1$ are in  $ P$, (iii) $\hat{x}^0_\ell=0$ and $\hat{x}^1_\ell=1$, and (iv) $\gamma_0 \hat{x}^0 + \gamma_1\hat{x}^1 \leq x'$.
\end{lemma}

\begin{proof}	
	Consider the following linear program which we denote by $\LPC(\ell,x')$. The variables of $\LPC(\ell,x')$ are $\lambda_0$, $\lambda_1$, $x^0$ and $x^1$.
	\begin{align}
	\LPC(\ell,x')\quad\quad& \max\quad \;\lambda_0+\lambda_1\\
	&\;\text{s.t.} \quad Ax^j \geq b\lambda_j & \mbox{ for $j=0,1$} \label{feasibility}\\
	&\;{\color{white}{\text{s.t.}} }\quad 0 \leq x^j \leq \lambda_j &\mbox{ for $j=0,1$}\label{bound}\\
	&\;{\color{white}{\text{s.t.}} }\quad x^0_\ell = 0,\; x^1_\ell =\lambda_1\label{branchcoordinate}\\
	&\;{\color{white}{\text{s.t.}} }\quad x^0 + x^1 \leq x'\label{packing}\\
	&\;{\color{white}{\text{s.t.}} }\quad \lambda_0,\lambda_1 \geq 0
	\end{align}
	
	Let $y^0$, $y^1$, $\lambda^*_0$, and $\lambda^*_1$ be an optimal solution to the LP above. Set $\gamma_0=\lambda^*_0$ and $\gamma_1 = \lambda^*_1$.
\cindy{The sentence right before this is how I think you link the $\lambda$ in the LP to the $\gamma$ constants in the Lemma statement and the claim below, etc.}
Let $\hat{x}^0 = y^0/\lambda^*_0$, $\hat{x}^1=y^1/\lambda^*_1$. This choice satisfies  (ii), (iii), (iv). To show that (i) is also satisfied we prove the following claim.
	\arash{Now that I notice, we use $\lambda_0,\lambda_1,x^0,x^1$ to represnt a feasible solution to the system 13-18. Since the optimal solution is use only for the duration of one sentence above maybe we can change the notation for the optimal solution. For instance use $y^0,y^1, \lambda^*_0,\lambda^*_1$}
\cindy{I think I made the changes you suggested.  I also used the LP optima to set the $\hat{x}$.  Please check.}
	\begin{claim}\label{CVexists}
		We have $\gamma_0 + \gamma_1\geq 1/g$.
	\end{claim}
	\begin{cproof}
		We show that there is a feasible solution that achieves the objective value of $\frac{1}{g}$. By Theorem \ref{CV2} there exists $\theta \in [0,1]^k$, with $\sum_{i=1}^{k}\theta_i = 1$ and $\tilde{x}^i\in S$ for $i\in[k]$ such that 
		$\sum_{i=1}^{k}\theta_i \tilde{x}^i\leq gx'$. So
		
		\begin{equation}\label{splitting}
		x'\geq \sum_{i=1}^{k}\frac{\theta_i}{g} \tilde{x}^i
		={\sum_{i\in [k]: \tilde{x}^i_\ell =0}\frac{\theta_i}{g} \tilde{x}^i}+{\sum_{i\in [k]: \tilde{x}^i_\ell =1}\frac{\theta_i}{g} \tilde{x}^i}.
		\end{equation}
		For $j=0,1$, let $x^j = \sum_{i\in [k]:\tilde{x}^i_\ell=j} \frac{\theta_i}{g}\tilde{x}^i$. Also let $\lambda_0=\sum_{i\in [k]: \tilde{x}^i_\ell =0}\frac{\theta_i}{g}$ and $\lambda_1 = \sum_{i\in [k]: \tilde{x}^i_\ell =1}\frac{\theta_i}{g}$. Note that $\lambda_0+\lambda_1 =1/g$. Constraint (\ref{packing}) is satisfied by Inequality (\ref{splitting}). Also, for $j=0,1$ we have
		\begin{equation}
		Ax^j= \sum_{i\in[k], \tilde{x}^i_\ell = j} \frac{\theta_i}{g} A\tilde{x}^i  \geq b \sum_{i\in[k], \tilde{x}^i_\ell = j} \frac{\theta_i}{g} = b\lambda_j.
		\end{equation}
		Hence, Constraints (\ref{feasibility}) holds. Constraint (\ref{branchcoordinate}) also holds since $x^0_\ell$ is obviously $0$ and $x^1_\ell= \sum_{i\in [k]: \tilde{x}^i_\ell = 1}\frac{\theta_i}{g}= \lambda_1$. The rest of the constraints trivially hold. 
	\end{cproof}
	This concludes the proof of Lemma \ref{LPClemma}.	
\end{proof}

We now show if $x'$ in the statement of Lemma \ref{LPClemma} is partially integral, we can find solutions with more integral components.
\begin{lemma}\label{round-up}
	Given $x'\in \dom(P)$ where $x'_1,\ldots,x'_{\ell-1}\in \{0,1\}$ and $x'_{\ell}<1$ for some $\ell\geq 1$ we can find in polynomial time vectors $\hat{x}^0,\hat{x}^1$ and scalars $\gamma_0,\gamma_1 \in [0,1]$ such that: (i) ${ 1}/{g}\leq \gamma_0 + \gamma_1  \leq 1$, (ii) $\hat{x}^0$ and $\hat{x}^1$ are in  $\dom( P)$, (iii) $\hat{x}^0_\ell=0$ and $\hat{x}^1_\ell=1$, (iv) $ \gamma_0\hat{x}^0 +\gamma_1 \hat{x}^1 \leq
	x'$,(v) $\hat{x}^i_j\in \{0,1\}$ for $i=0,1$ and $j\in[\ell-1]$.
\end{lemma} 
\begin{proof}
	By Lemma \ref{LPClemma} we can find $\bar{x}^0$, $\bar{x}^1$, $\gamma_0$ and $\gamma_1$ that satisfy (i), (ii), (iii), and (iv). We define $\hat{x}^0$ and $\hat{x}^1$ as follows. For $i=0,1$, for $j\in[\ell-1]$, let $\hat{x}^i_j= \ceil{\bar{x}^i_j}$, for $j=\ell,\ldots,t$ let $\hat{x}^i_j = \bar{x}^i_j$.

	We now show that $\hat{x}^0$, $\hat{x}^1$, $\gamma_0$, and $\gamma_1$ satisfy all the conditions. Note that conditions (i), (ii), (iii), and (v) are trivially satisfied. Thus we only need to show (iv) holds. We need to show that $\gamma_0 \hat{x}^0_j+\gamma_1\hat{x}^1_j\leq gx'_j$. If $j=\ell,\ldots,t$, then this clearly holds. Hence, assume $j\leq \ell-1$. By the property of $x'$ we have $x'_j\in \{0,1\}$. If $x'_j= 0$, then by Constraint (\ref{packing}) we have $\bar{x}^0_j = \bar{x}^1_j=0$. Therefore, $\hat{x}^i_j=0$ for $i=0,1$, so (iv) holds. Otherwise if $x'_j = 1$, then we have
	$\gamma_0\hat{x}^0_j+\gamma_1\hat{x}^1_j\leq \gamma_0+\gamma_1\leq 1\leq x'_j.$ 
	Therefore (v) holds.
\end{proof}

\paragraph{Growing and Pruning FDT tree.} The FDT algorithm maintains nodes $L_i$ in iteration $i$ of the algorithm. The nodes in $L_i$ correspond to the nodes in level $L_i$ of the FDT tree. The points in the leaves of the FDT tree, $L_t$, are points in $\dom(P)$ and are integral for all integer variables.

\begin{lemma}\label{prune}
	There is a polynomial time algorithm that produces sets $L_0,\ldots,L_t$ of pairs of $x\in \dom(P)$ together with multipliers $\lambda$ with the following properties for $i=0,\ldots,t$:
	(a) If $x\in L_i$, then $x_j \in \{0,1\}$ for $j\in [i]$, i.e. the first $i$ coordinates of a solution in level $i$ are integral, (b) $\sum_{[x,\lambda]\in L_i} \lambda\geq\frac{1}{g^i}$, (c) $\sum_{[x,\lambda]\in L_i}\lambda x \leq x^*$, (d) $|L_i|\leq t$.
\end{lemma}
\begin{proof}
	We prove this lemma using induction but one can clearly see how to turn this proof into a polynomial time algorithm. Let $L_0$ be the set that contains a single node (\textit{root of the FDT tree}) with $x^*$ and multiplier 1. All the requirements in the lemma are satisfied for this choice.
	
	Suppose by induction that we have constructed sets $L_0,\ldots,L_i$. Let the solutions in $L_i$ be $x^j$ for $j\in [k]$ and $\lambda_j$ be their multipliers, respectively. For each $j\in[k]$ if $x^j_{i+1}=1$ we add the pair $(x^j,\lambda_j)$ to $L'$. Otherwise,	applying Lemma \ref{round-up} (setting $x'= x^j$ and $\ell = i+1$) we can find $x^{j0}$, $x^{j1}$, $\lambda^0_j$ and $\lambda^1_j$ with the properties (i) to (v) in Lemma \ref{round-up}. Add the pairs  $(x^{j0} ,\lambda_j\lambda^0_j)$ and  $(x^{j1} ,\lambda_j\lambda^1_j)$ to $L'$. It is easy to check that set $L'$ is a suitable candidate for $L_{i+1}$, i.e. set $L'$ satisfies (a), (b) and (c). However we can only ensure that $|L'|\leq 2k\leq 2t$, and might have $|L'|>t$. We call the following linear program $\prun(L')$. Let $L' = \{[x^1,\gamma_1],\ldots,[x^{|L'|},\gamma_{|L'|}]\}$. The variables of $\prun(L')$ are scalar variables $\theta_j$ for each node $j$ in $L'$.  
		\begin{equation}
		\prun(L')\quad\quad\quad \{\max \sum_{j=1}^{|L'|} \theta_j\;:\; \sum_{j=1}^{|L'|} \theta_j x^j_i\leq x^*_i \mbox{ for $i\in [t]$},\; \theta\geq 0 \}
		\end{equation}
		
		Setting $\theta = \gamma$ gives a feasible solution to $\prun(L')$. Let $\theta^*$ be the optimal vertex solution to this LP. Since the problem is in $\mathbb{R}^{|L'|}$,  $\theta^*$ has to satisfy $|L'|$ linearly independent constraints at equality. However, there are only $t$ constraints of type $ \sum_{j=1}^{|L'|} \theta_j x^j_i\leq x^*_i$. Therefore, there are at most $t$ coordinates of $\theta^*_j$ that are non-zero. Set $L_{i+1}$ which consists of $x^j$ for $j=1,\ldots,|L'|$ and their corresponding multipliers $\theta^*_j$ satisfy the properties in the statement of the lemma. We can discard the nodes in $L_{i+1}$ that have $\theta^*_j=0$, so $|L_{i+1}| \leq t$. Also, since $\theta^*$ is optimal and $\gamma$ is feasible for $\prun(L')$, we have $\sum_{j=1}^{|L'|} \theta^*_j \geq \sum_{j=1}^{|L'|}\gamma_j \geq \frac{1}{g^{i+1}}$. \end{proof}
	
	\paragraph{From leaves of FDT to feasible solutions.}
	For the leaves of the FDT tree,  $L_t$, we have that every solution $x$ in $L_t$ has $x\in\{0,1\}^n$ and $x\in \dom(P)$. By applying Lemma \ref{domlemma} we can obtain a point $x'\in S$ such that $x'\leq x$. This concludes the description of the FDT algorithm and proves Theorem \ref{binaryFDT}. See Algorithm \ref{FDTFull} for a summary of the FDT algorithm.
	
	\vspace*{8pt}

	\begin{algorithm}[H]\label{FDTFull}
		\KwIn{$P= \{x\in \mathbb{R}^{n}: Ax\geq b\}$ and $S=\{x\in P: x\in \{0,1\}^n\}$ such that $g=\max_{c\in \mathbb{R}^n_+ }\frac{\min_{x\in S}cx}{\min_{x\in P}cx}$ is finite, $x^*\in P$}
		\KwOut{$z^i\in S$ and $\lambda_i\geq 0$ for $i\in[k]$ such that $\sum_{i=1}^{k}\lambda_i = 1$, and $\sum_{i=1}^{k}\lambda_iz^i\leq g^tx^*$ }
		$L^0\leftarrow [x^*,1]$\\
		\For{$i=1$ \textbf{to} $t$}{
			$L'\leftarrow \emptyset$\\
			\For{$[x,\lambda] \in L^i$}{
				Apply Lemma \ref{round-up} to obtain $[\hat{x}^0,\gamma_0]$ and $[\hat{x}^1,\gamma_1]$\\
				$L' \leftarrow L' \cup \{[\hat{x}^0,\lambda\cdot\gamma_0]\} \cup \{[\hat{x}^1,\lambda\cdot\gamma_1]\}$\\			
			}
			Apply Lemma \ref{prune} to prune $L'$ to obtain $L^{i+1}$. 
		}
		\For{$[x,\lambda] \in L^t$}{
			Apply Algorithm \ref{domtoIPalg} to $x$ to obtain $z\in S$\\
			$F \leftarrow F \cup \{[z,\lambda]\}$
		}
		\textbf{return} $F$
		\caption{Fractional Decomposition Tree Algorithm}
	\end{algorithm}

\vspace{8pt}
There are $O(n^2)$ nodes in the FDT tree. A faster way to achieve feasible solutions with good quality for an IP with bounded integrality gap is an algorithm that takes a random dive into the FDT tree, hence only visiting $O(n)$ nodes.

	\vspace*{8pt}
	
	\begin{algorithm}[H]\label{FDT-dive}
		\KwIn{$P= \{x\in \mathbb{R}^{n}: Ax\geq b\}$ and $S=\{x\in P: x\in \{0,1\}^n\}$ such that $g=\max_{c\in \mathbb{R}^n_+ }\frac{\min_{x\in S}cx}{\min_{x\in P}cx}$ is finite, $x^*\in P$}
		\KwOut{$z \in S$}
		$y = x^*$\\
		\For{$i=1$ \textbf{to} $t$}{
			
			Apply Lemma \ref{round-up} to obtain $[\hat{x}^0,\gamma_0]$ and $[\hat{x}^1,\gamma_1]$\\
			$i \sim\text{Bernoulli}(\frac{\gamma_0}{\gamma_0+\gamma_1})$\\
			$y\rightarrow \hat{x}^i$
			
		}
		Apply Algorithm \ref{domtoIPalg} to $y$ to obtain $z\in S$\\
		\textbf{return} $z$
		\caption{Dive FDT Algorithm}
	\end{algorithm}

\section{FDT for 2EC}\label{sec:2EC}

In Section~\ref{sec:binaryfdt} our focus was on binary IPs. In this section, in an attempt to extend FDT to \{0,1,2\} problems we introduce an FDT algorithm for a 2-edge-connected multi-subgraph problem. Given a graph $G=(V,E)$ a multi-subset of edges $F$ of $G$ is a 2-edge-connected multi-subgraph of $G$ if for each set $\emptyset\subset U \subset V$, the number of edges in $F$ that have one endpoint in $U$ and one not in $U$ is at least 2. In 2EC, we are given non-negative costs on the edges of $G$ and the goal is to find the minimum cost 2-edge-connected multi-subgraph of $G$. We want to prove Theorem \ref{FDT2EC}.
\FDTEC*
We do not know the exact value for $g(\2ec)$, but we know $\frac{6}{5} \leq g(\2ec) \leq \frac{3}{2}$ \cite{alexander2006integrality,wolsey}. The FDT algorithm for 2EC is very similar to the one for binary IPs, but there are some differences as well. A natural thing to do is to have three branches for each node of the FDT tree, however, the branches that are equivalent to setting a variable to $1$, might need further decomposition. That is the main difficulty when dealing with $\{0,1,2\}$-IPs.

First, we need a branching lemma. Observe that  the following branching lemma is essentially a translation of Lemma \ref{LPClemma} for $\{0,1,2\}$ problems except for one additional clause. 

\begin{restatable}{lemma}{2ECLPC}
	\label{LPC2EC}
	Given $x\in \subtour(G)$, and $e\in E$ we can find in polynomial time vectors $x^0,x^1$ and $x^2$ and scalars $\gamma_0,\gamma_1$, and $\gamma_2$ such that: (i) $\gamma_0 + \gamma_1 +\gamma_2 \geq { 1}/{g(\2ec)}$, (ii) $x^0,x^1,$ and $x^2$ are in  $ \subtour(G)$, (iii) $x^0_e=0$, $x^1_e=1$, and $x^2_e=2$, (iv) $\gamma_0 x^0 + \gamma_1{x}^1  + \gamma_2x^2\leq {x}$, (v) for $f\in E$ with ${x}_f\geq 1$, we have $x^j_f\geq 1$ for $j=0,1,2$.
\end{restatable}

\begin{proof}
	Consider the following LP with variables $\lambda_j$ and $x^j$ for $j=0,1,2$. 
	\begin{align}
	\quad\quad& \max\quad \;\sum_{j=0,1,2}\lambda_j\\
	&\;\text{s.t.} \quad x^j(\delta(U))\geq 2\lambda_j \;& \mbox{ for $\emptyset \subset U \subset V$, and $j=0,1,2$} \label{feasibility2ec}\\
	&\;{\color{white}{\text{s.t.}} }\quad 0 \leq x^j \leq 2\lambda_j\; &\mbox{ for $j=0,1,2$}\label{bound2ec}\\
	&\;{\color{white}{\text{s.t.}} }\quad x^j_e =j\cdot \lambda_j\; &\mbox{ for $j=0,1,2$}\label{branchcoordinate2ec}\\
	&\;{\color{white}{\text{s.t.}} }\quad x^j_f \geq \lambda_j \; &\mbox{ for $f\in E$ where $x_f \geq 1$, and $j=0,1,2$}\label{1edges2ec}\\
	&\;{\color{white}{\text{s.t.}} }\quad x^0 + x^1+x^2 \leq x\label{packing2ec}\\
	&\;{\color{white}{\text{s.t.}} }\quad \lambda_0,\lambda_1,\lambda_2 \geq 0
	\end{align}	Let $x^j$, $\gamma_j$ for $j=0,1,2$ be an optimal solution solution to the LP above. Let $\hat{x}^{j}={x^j}/{\gamma_j}$ for $j=0,1,2$ where $\gamma_j>0$. If $\gamma_j=0$, let $\hat{x}^{j}=0$. Observe that  (ii), (iii), (iv), and (v) are satisfied with this choice. We can also show that $\gamma_0+\gamma_1+\gamma_2\geq {1}/{g(\2ec)}$, which means that (i) is also satisfied. The proof is similar to the proof of the claim in Lemma \ref{LPClemma}, but we need to replace each $f\in E$ with $x_f\geq 1$ with a suitably long path to ensure that Constraint (\ref{1edges2ec}) is also satisfied.	
	\begin{claim}\label{CVexists}
		We have $\gamma_0 + \gamma_1+\gamma_2\geq \frac{1}{g(\2ec)}$.
	\end{claim}
	\begin{cproof}
		Suppose for contradiction $\sum_{j=0,1,2}\gamma_j = \frac{1}{g(\2ec)} - \epsilon$ for some $\epsilon >0$. Construct graph $G'$ by removing edge $f$ with $x_f\geq 1$ and replacing it with a path $P_f$ of length $\ceil{\frac{2}{\epsilon}}$. Define $x'_h = x_h$ for each edge $h$ such that $x_h<1$. For each $h\in P_f$ let $x'_h= x_f$ for all $f$ with $x_f\geq 1$. It is easy to check that $x'\in \subtour(G')$. By Theorem \ref{CV2} there exists $\theta \in [0,1]^k$, with $\sum_{i=1}^{k}\theta_i = 1$ and 2-edge-connected multi-subgraphs $F'_i$ of $G'$ for $i=1,\ldots,k$ such that 
		$\sum_{i=1}^{k}\theta_i \chi^{F'_i}\leq g(\2ec)x'$. 
		
		Each $F'_i$ contains at least one copy of every edge in any path $P_f$, except for at most one edge in the path. We will obtain 2-edge-connected multi-subgraphs $F_1,\ldots,F_k$ of $G$ using $F'_1,\ldots,F'_k$, respectively. To obtain $F_i$ first remove all $P_f$ paths from $F'_i$. Suppose there is an edge $h$ in $P_f$ such that $\chi^{F'_i}_h=0$, this means that for any edge $p\in P_f$ such that $p\neq h$, $\chi^{F'_i}_p=2$. In this case, let $\chi^{F_i}_f=2$, i.e. add two copies of $f$ to $F_i$. If there are at least one edge $h\in P_f$ with $\chi^{F'_i}_h= 1$, let $\chi^{F_i}_f=1$, i.e. add one copy of $f$ to $F_i$. If for all edges $h\in P_f$, we have $\chi^{F'_i}_h=2$, then let $\chi^{F_i}_f=2$. For $f\in E$ with $x_f<1$ we have
		\begin{equation}
		\sum_{i=1}^{k}\theta_i \chi^{F_i}_f=\sum_{i=1}^{k}\theta_i \chi^{F'_i}_f\leq g(\2ec)x'_f= g(\2ec)x_f.
		\end{equation}
		In addition for $f\in E$ with $x_f\geq 1$ we have $\chi^{F_i}_f \leq \frac{\sum_{h\in P_f}\chi^{F'_i}_h}{\ceil{\frac{2}{\epsilon}}-1}$ by construction.
		\begin{align*}
		\sum_{i=1}^{k}\theta_i \chi^{F_i}_f&\leq \sum_{i=1}^{k}\theta_i\frac{\sum_{h\in P_f}\chi^{F'_i}_h}{\ceil{\frac{2}{\epsilon}}-1}\\
		&= \frac{\sum_{h\in P_f} \sum_{i=1}^{k}\theta_i\chi^{F'_i}_h}{\ceil{\frac{2}{\epsilon}}-1}\\
		&\leq \frac{\sum_{h\in P_f} g(\2ec)x'_h}{\ceil{\frac{2}{\epsilon}}-1}\\
		&= \frac{\sum_{h\in P_f} g(\2ec)x_f}{\ceil{\frac{2}{\epsilon}}-1}\\
		&= \frac{\ceil{\frac{2}{\epsilon}}}{\ceil{\frac{2}{\epsilon}}-1}g(\2ec)x_f.
		\end{align*}
		Therefore, since $\frac{\ceil{\frac{2}{\epsilon}}}{\ceil{\frac{2}{\epsilon}}-1}\geq 1$, we have 
		\begin{equation}
		x \geq\sum_{i\in [k]: \chi^{F_i}_e=0}\frac{\theta_i(\ceil{\frac{2}{\epsilon}}-1)}{g(\2ec)\ceil{\frac{2}{\epsilon}}}\chi^{F_i}+ \sum_{i\in [k]: \chi^{F_i}_e=1}\frac{\theta_i(\ceil{\frac{2}{\epsilon}}-1)}{g(\2ec)\ceil{\frac{2}{\epsilon}}}\chi^{F_i}+\sum_{i\in [k]: \chi^{F_i}_e=2}\frac{\theta_i(\ceil{\frac{2}{\epsilon}}-1)}{g(\2ec)\ceil{\frac{2}{\epsilon}}}\chi^{F_i}.
		\end{equation}
		Let $x^j = \sum_{i\in [k]: \chi^{F_i}_e=j}\frac{\theta_i(\ceil{\frac{2}{\epsilon}}-1)}{g(\2ec)\ceil{\frac{2}{\epsilon}}}\chi^{F_i}$ and $\theta_j =  \sum_{i\in [k]: \chi^{F_i}_e=j}\frac{\theta_i(\ceil{\frac{2}{\epsilon}}-1)}{g(\2ec)\ceil{\frac{2}{\epsilon}}}$ for $j=0,1,2$. It is easy to check that $x^j$ , $\theta_j$ for $j=0,1,2$ is a feasible solution to the LP above. We have $\sum_{j=0,1,2}\theta_j = \frac{\ceil{\frac{2}{\epsilon}}-1}{g(\2ec)\ceil{\frac{2}{\epsilon}}}$. By assumption, we have $\frac{\ceil{\frac{2}{\epsilon}}-1}{g(\2ec)\ceil{\frac{2}{\epsilon}}}\leq  \frac{1}{g(\2ec)}-\epsilon$, which is a contradiction.
	\end{cproof}
	This concludes the proof. \end{proof}
In contrast to FDT for binary IPs where we round up the fractional variables that are already branched on at each level, in FDT for 2EC we keep all coordinates as they are and perform a rounding procedure at the end. Formally, let $L_i$ for $i=1,\ldots,|\spp(x^*)|$ be collections of pairs of feasible points in $\subtour(G)$ together with their multipliers. Let $t=|\spp(x^*)|$ and assume without loss of generality that $\spp(x^*)=\{e_1,\ldots,e_t\}$. 

\begin{lemma}\label{2ecpruning}
	The FDT algorithm for 2EC in  polynomial time produces sets $L_0,\ldots,L_t$ of pairs $x\in \2ec(G)$ together with multipliers $\lambda$ with the following properties for $i\in [t]$:
	(a) If $x\in L_i$, then $x_{e_j}=0$ or $x_{e_j}\geq 1$ for $j=1,\ldots,i$, (b) $\sum_{(x,\lambda)\in L_i }\lambda \geq \frac{1}{g(\2ec)^i}$, (c) $\sum_{(x,\lambda)\in L_i }\lambda x \leq x^*$, (d) $|L_i|\leq t$.
\end{lemma}
The proof is similar to Lemma \ref{prune}, but we need to use property (v) in Lemma \ref{LPC2EC} to prove that (a) also holds.
\begin{proof}
	We proceed by induction on $i$. Define $L_0=\{(x^*,1)\}$. It is easy to check all the properties are satisfied. Now, suppose by induction we have $L_{i-1}$ for some $i=1,\ldots,t$ that satisfies all the properties. For each solution $x^\ell$ in $L_{i-1}$ apply Lemma \ref{LPC2EC} on $x^\ell$ and $e_{i}$ to obtain $x^{\ell j}$ and $\lambda_{\ell j}$ for $j=0,1,2$. Let $L'$ be the collection that contains $(x^{\ell j},\lambda_\ell \cdot \lambda_{\ell j})$ for $j=0,1,2$, when applied to all $(x^\ell,\lambda_\ell)$ in $L_{i-1}$. Similar to the proof in Lemma \ref{prune} one can check that $L_i$ satisfies properties (b), (c). We now verify property (a). Consider a solution $x^\ell$ in $L_{i-1}$. For $e\in \{e_1,\ldots,e_{i-1}\}$ if $x^\ell_e =0$, then by property (iv) in Lemma \ref{LPC2EC} we have $x^{\ell j}=0$ for $j=0,1,2$. Otherwise by induction we have $x^{\ell}_{e}\geq 1$ in which case property (v) in Lemma \ref{LPC2EC} ensures that $x^{\ell j}_e\geq 1$ for $j=0,1,2$. Also, $x^{\ell j}_{e_i}= j$, so $x^{\ell j}_{e_i}=0$ or $x^{\ell j}_{e_i}\geq 1$ for $j=0,1,2$. 
	
	Finally, if $|L'|\leq t$ we let $L_i=L'$, otherwise apply $\prun(L')$ to obtain $L_{i}$.
\end{proof}

Consider the solutions $x$ in $L_t$. For each variable $e$ we have $x_e=0$ or $x_e\geq 1$. 
\begin{lemma}\label{rounddown}
	Let $x$ be a solution in $L_t$. Then $\floor{x} \in \subtour(G)$. 
\end{lemma}
\begin{proof}
	Suppose not. Then there is a set of vertices $\emptyset \subset U \subset V$ such that $\sum_{e\in \delta(U)}\floor{x_e}<2$. Since $x\in \subtour(G)$ we have $\sum_{e\in \delta(U)}x_e \geq 2$. Therefore, there is an edge $f\in \delta(U)$ such that $x_f$ is fractional. By property (a) in Lemma \ref{2ecpruning}, we have $1<  x_f < 2$. Therefore, there is another edge $h$ in $\delta(U)$ such that $x_h>0$, which implies that $x_h\geq 1$. But in this case $\sum_{e\in \delta(U)}\floor{x_e}\geq  \floor{x_f}+\floor{x_h}  \geq 2$. This is a contradiction.
\end{proof}

The FDT algorithm for 2EC iteratively applies Lemmas \ref{LPC2EC} and \ref{2ecpruning} to variables $x_1,\ldots,x_t$ to obtain leaf point solutions $L_t$. Finally, we just need to apply Lemma \ref{rounddown} to obtain the 2-edge-connected multi-subgraphs from every solution in $L_t$. Since $x$ is an extreme point we have $t\leq 2|V|-1$ \cite{boydpulley}. By Lemma \ref{2ecpruning} we have
\begin{align*}
\sum_{(x,\lambda)\in L_t} \frac{\lambda}{\sum_{(x,\lambda)\in L_t}\lambda} \floor{x} \leq \frac{1}{\sum_{(x,\lambda)\in L_t}\lambda} \sum_{(x,\lambda)\in L_t} \lambda {x} \leq g^t_{\2ec} x^*.
\end{align*}
\section{Computational Experiments with FDT}\label{sec:experiment}
We ran FDT on three network design problems: VC, TAP and 2EC. 

We implemented the experiments for VC and TAP in Python running on a linux workstation (Ubuntu 18.04.3) with 8 cores of Intel(R) Core(TM) i7-8565U CPU  1.80GHz processors and 1Mb of cache. We used the CPLEX 12.9.0.0 solver to solve the pyomo LP models. We ran the experiments for 2EC on a Windows machine, coded in AMPL with CPLEX as the solver.
\paragraph{FDT on VC instances from (PACE 2019) \cite{PACE}.}

We compared Dive FDT (Algorithm \ref{FDT-dive}) with feasbility pump \cite{fp1} in terms of running time
and the quality of solution provided by each algorithm. We used the small (200 vertex) test cases from the PACE 2019 vertex-cover challenge. The results are presented in Figure \ref{fpvsfdt}. 

\begin{figure}[h!]
\begin{subfigure}{.5\textwidth}
\centering
	\includegraphics[width=8cm,scale=1]{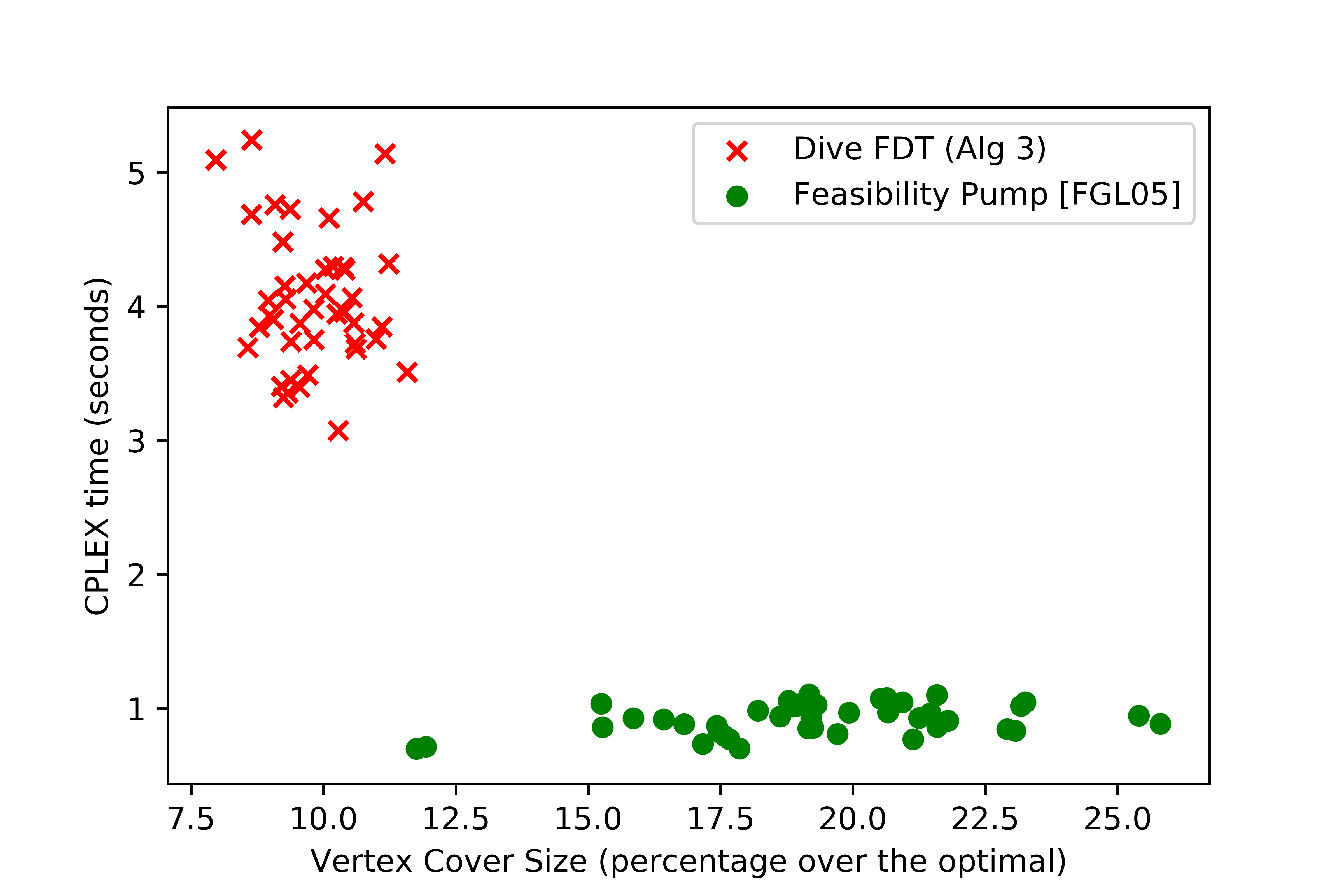}
	\caption{Dive FDT vs feasbility pump on the instances of PACE 2019 \cite{PACE} with 200 vertices.}
	\label{fpvsfdt}
	\end{subfigure}
	$\quad\;$
	\begin{subfigure}{.5\textwidth}
	\centering
	\includegraphics[width=8cm,scale=1]{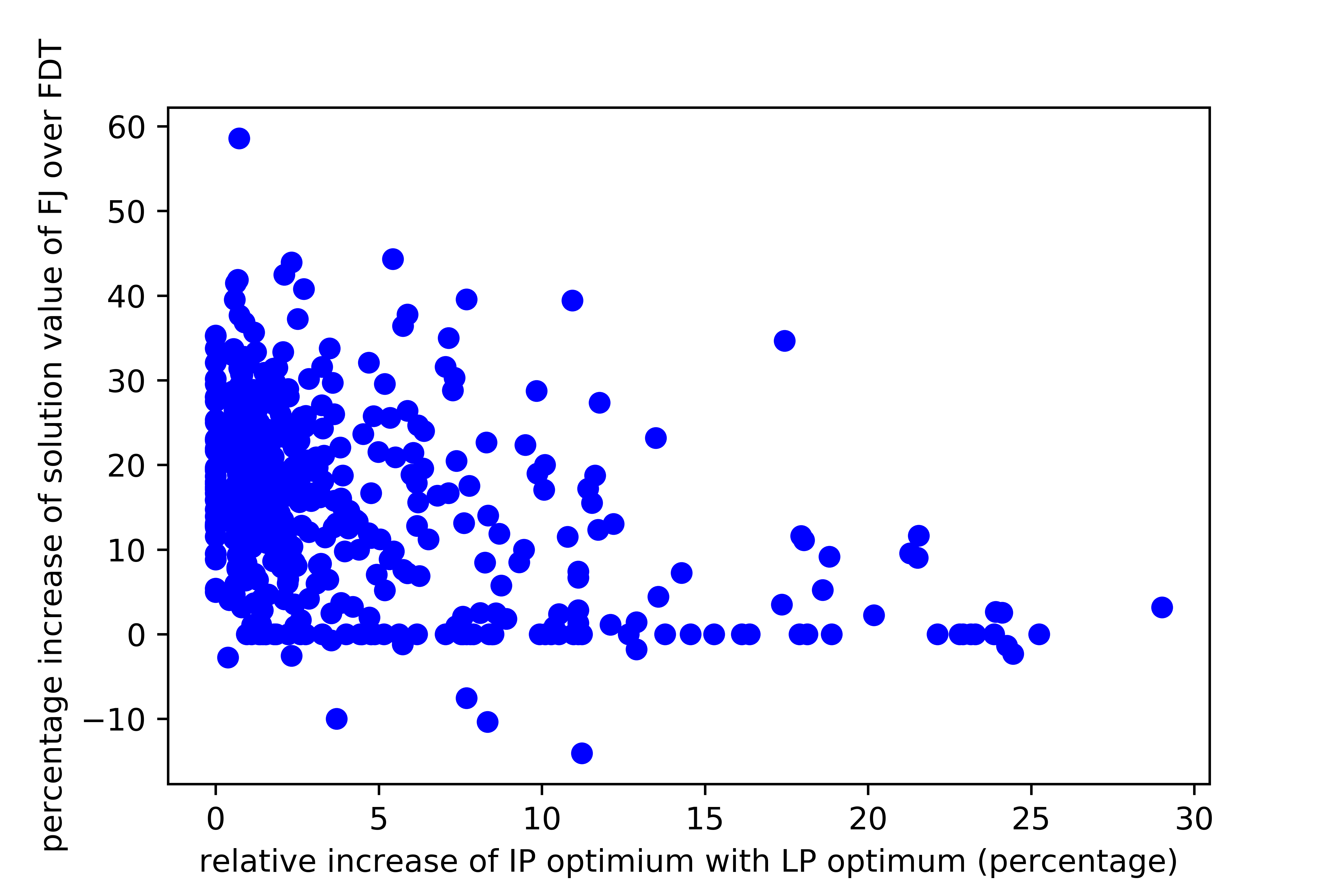}
	\caption{TAP on our random instances: FDT run on LP optimal vs the 2-approximation for TAP\cite{FJ81}.}
	\label{fjvsfdt}
	\end{subfigure}
	\caption{Computational experiments with the FDT algorithm}
	\label{fdtcomp}
\end{figure}
\paragraph{FDT on randomly generated instances of TAP.}
Recall that in the tree augmentation problem (TAP) we are given a tree $T=(V,E)$, a set of non-tree links $L$ between vertices in $V$ and costs $c\in \mathbb{R}^{L}_{\geq 0}$. A feasible augmentation is $L'\subseteq L$ such that $T+L'$ is 2-edge-connected. In TAP we wish to find the minimum-cost feasible augmentation. The integrality gap of the cut LP for TAP (given in (\ref{eq:cutLP})) is
\begin{equation*}
g(\tap) = \max_{c\in \mathbb{R}^L_{\geq 0}} \frac{\min_{x\in\tap(T,L)} cx}{\min_{x\in\cut(T,L)} cx},
\end{equation*}  
where $\tap(T,L)$ is the feasible set for the IP and $\cut(T,L)$ is the feasible set for the cut LP (relaxation).
We know $\frac{3}{2}\leq g(\tap)\leq 2$~\cite{FJ81,32gaptap}. The IP $\min_{x\in \tap(T,L)}cx$ is binary. 

As input for our experiments, we considered full binary trees with 3 to 7 levels with a link for each pairs of leaves. We set the link costs uniformly at random.  We summarize these test instances in Table~\ref{tableTAP}. We ran binary FDT on each test instance and chose the solution with minimum cost. We compare the FDT solutions to those from the circulation-based 2-approximation algorithm of Frederickson and J\'{a}J\'{a}~\cite{FJ81} in Figure \ref{fjvsfdt}.  The simple FJ heuristic was faster than FDT, from about 8x to about 60x on our examples. However, FDT still ran in less than a minute on the bigger problems and less than 10 seconds on the others. FDT (in the worst case) solves $\Omega(m^2)$ LPs, where $m$ is the number of edges, and FJ solves one LP. So the difference in running time, which depends on instance size, is not surprising. FDT runs in polynomial time and usually gives much better solutions than FJ. 

\begin{table}[h!]
	\begin{small}
		\centering
		\begin{tabular}{c c c c}
			\hline
			& number of edges in $T$ & number of links in $L$ & number of instances $(T,L)$\\ \hline
			 & $6$ & $6$ & $100$ \\ 
			 & $14$ & $28$ & $100$ \\ 
			 & $30$ & $120$ & $100$ \\
			 & $62$ & $496$ & $100$ \\ 
			  & $126$ & $2016$ & $100$ \\  \hline \\
		\end{tabular}\caption{Summary of the randomly generated instances of TAP.} \label{tableTAP}
	\end{small}
\end{table}
For all 500 instances in our experiments, running FDT on the LP optimal (fractional extreme point) of the cut LP gave a feasible solution with value at most a factor $\frac{3}{2}$ larger than the LP lower bound.  Such a feasible solution gives an upper bound on the integrality gap $g(\tap)$ of that specific instance of at most $\frac{3}{2}$.
In fact, the integrality gap upper bound derived this way was equal to $\frac{3}{2}$ for only one instance.
For 480 instances, the integrality gap upper bound was $\frac{4}{3}$, for 16 instances it was $\frac{6}{5}$, for 2 instances it was $\frac{8}{7}$, for 1 instance it was $\frac{10}{9}$. 
\paragraph{Computational comparison between Christofides' algorithm and FDT for 2EC on Carr-Vempala points.} 

As described in Section~\ref{sec:2EC-intro}, when the LP optimum for (\ref{eq:subtour}) is a Carr-Vempala point (see Figure~\ref{fig:CVpoint}), the integrality gap of that instance is equal to that of 2EC. So in that sense, these are the hardest fractional points to round to a feasible solution. We now compare the quality of FDT's solution from a Carr-Vempala point to that of Christofides, the best known approximation algorithm.  In the classic graph-based Christofides algorithm for TSP, find a minimum spanning tree of the graph.  Then, since TSP runs on a complete graph, add a perfect matching on the vertices that have odd degree in the spanning tree to make the (multi)-subgraph 2-edge-connected. If we stop that algorithm without shortcutting, we have a feasible solution for Metric-2EC (given a complete graph, find the minimum-cost 2-edge-connected subgraph, with no multiedges).
\arash{The following sentence is not accurate, Metric-TSP is the same as TSP on a non-complete graph where we can take multiple copies of the edges, Metric-2EC (given complete graph with metric costs find the minimum cost 2-edge-connected SUBgraph -no doubling of edges-) is the same as 2EC on a non-complete graph where we can take multiple copies of edges.}
\cindy{It sounded like 2EC doesn't assume a complete graph.  But I removed the sentence.  Is the following OK?}
\arash{Still not quite accurate. The sentence gives the impression that changing matching with $T$-join is the difference between TSP and 2EC, but its the difference between working with complete graph with metric costs and a non-complete graph where doubling edges is allowed.}
For 2EC, instead of a matching, we add an $O$-join, where $O$ is the set of odd-degree vertices in the spanning tree~\footnote{For graph $G=(V,E)$ and $O\subseteq V$ with $|O|$ even, an $O$-join of $G$ is a subgraph of $G$ that has odd degree on the vertices in $O$ and even degree on vertices in $V\setminus O$. Graphs always have an even number of odd vertices since every edge has two endpoints, giving an even number of endpoints.}. The best-of-many Christofides (BOMC) algorithm~\cite{AKS15} works better than the classic Christofides algorithm in practice~\cite{GW17}.  The BOMC algorithm samples many spanning trees from the LP relaxation (with edges chosen proportionally to the LP value), augments each spanning tree with the matching, and takes the best solution.


The polyhedral analysis of Christofides shows that if $x\in \subtour(G)$, then $\frac{3}{2}x$ dominates a convex combination of connected Eulerian (hence 2-edge-connected) multi-subgraphs of $G$~\cite{wolsey,shmoyswilliamson}. This bound holds for any point $x\in \subtour(G)$ for any graph $G$, including any Carr-Vempala point. We show how this bound can be improved for specific instances by doing a slightly improved analysis of the algorithm.

Let $x$ be a Carr-Vempala point (from~(\ref{eq:subtour})) defined on a graph $G=(V,E)$. 
If $x\in \subtour(G)$, then $\frac{|V|-1}{|V|}x$ can be written as a convex combination of spanning trees of $G$ (See Proposition 2 in~\cite{Vygen12} for instance).  More explicitly,
\begin{equation}
\label{eq:tree-decomp}
\frac{|V|-1}{|V|}\cdot x= \sum_{i=1}^{k}\lambda_i\chi^{T_i},
\end{equation}
where $T_i$ is spanning tree of $G$, $\chi^{T_i}$ is its incidence vector, $\sum_{i=1}^{k}\lambda_i=1$, and $\lambda_i\geq 0$ for $i\in [k]$. The spanning-tree polytope is integral so there are polynomial-time algorithms to construct this decomposition (See Proposition~\ref{cara} and surrounding discussion).

Let $O_i$ be the set of odd-degree vertices of spanning tree $T_i$.
Adding an $O_i$-join to tree $T_i$ gives a solution to 2EC.
We solve the following LP that allows us to find edges that simultaneously augment all the spanning trees to solutions to 2EC.
\begin{equation}\label{ojoinaverage}
\min \{ \alpha:\;\sum_{i=1}^{k} \lambda_i y^i = \alpha \cdot x,\;  y^i \in \dom(O_i\join(G)) \; \mbox{for $i\in [k]$}\}.
\end{equation}
The variables in the above LP are $y^i\in \mathbb{R}^{E}_{\geq 0}$ for $i\in [k]$, where $E$ is the number of edges in the graph.  The parameters $\lambda_i$ are those we computed in~(\ref{eq:tree-decomp}). For each $i\in [k]$, we require $y^i$ to be in the dominant of $O_i\join(G)$.  We enforce this by adding a set of constraints for each $i$ from~\cite[p. 490]{schrijverbook}. These constraints force the inclusion of at least one edge from the cut associated with each set of vertices that have an even number of elements from $O_i$. The  $y^i$ need not be integer vectors. Because the $O_i$-join polytope is integral, we can replace $y^i$ with a convex combination of $O_i$-joins of $G$ (proposition~\ref{cara}). Setting $y^i = \frac{x}{2}$ and $\alpha = \frac{1}{2}$ gives a feasible solution to this LP, yielding an integrality gap of at most $\frac{3}{2}$ for the subtour-elimination relaxation for the TSP \cite{wolsey}. More generally, this method gives a ($\frac{|V|-1}{|V|}+\alpha$)-approximation for the specific instance $x$. 

\cindy{Arash, I liked your explanation about how a single edge will need to contribute 3/2 in the classic method.  You can add that if you'd like, but that might not be necessary given the time.}

Figure \ref{fdtvschris} shows FDT's solutions on all Carr-Vempala points that have 10 vertices on the cycle formed by fractional edges. We show for these points the apporoximation factor provided by FDT is always better than those from the polyhedral version of Christofides' algorithm. In Figure~\ref{fdtvschris} the horizontal axis of the plot is indexed with the 60 Carr-Vempala points that we considered. For each Carr-Vempala point $x$, there are two data points. The value of the first data point depicted by a circle on the vertical axis is $\frac{|V|-1}{|V|}+\alpha$  where $\alpha$ is the optimal solution to (\ref{ojoinaverage}).
The value of the second data point depicted by a cross on the vertical axis is $C$ where $C$ is obtained from applying Theorem \ref{FDT2EC} to $x$. In other words, Figure \ref{fdtvschris} is comparing the instance-specific upper bound on integrality gap certified by Christofides' algorithm to the approximation factor of the FDT algorithm for 2EC.

\begin{figure}
\centering
\begin{tikzpicture}[scale=0.4]

\draw [dashed] [black, line width=0.2mm] plot [smooth, tension=0] coordinates {(0,4) (0,-4)};
\draw [dashed] [black, line width=0.2mm] plot [smooth, tension=0] coordinates {(2.83,2.83) (-4,0)};
\draw [dashed] [black, line width=0.2mm] plot [smooth, tension=0] coordinates {(4,0) (-2.83,-2.83)};
\draw [dashed] [black, line width=0.2mm] plot [smooth, tension=0] coordinates {(-2.83,2.83) (2.83,-2.83)};

\draw [-] [black, line width=0.2mm] plot [smooth, tension=0] coordinates {(0,4) (2.83,2.83)};

\draw [-] [black, line width=0.2mm] plot [smooth, tension=0] coordinates {(4,0) (2.83,2.83)};

\draw [-] [black, line width=0.2mm] plot [smooth, tension=0] coordinates {(4,0) (2.83,-2.83)};

\draw [-] [black, line width=0.2mm] plot [smooth, tension=0] coordinates {(0,-4) (2.83,-2.83)};

\draw [-] [black, line width=0.2mm] plot [smooth, tension=0] coordinates {(0,-4) (-2.83,-2.83)};

\draw [-] [black, line width=0.2mm] plot [smooth, tension=0] coordinates {(-4,0) (-2.83,-2.83)};

\draw [-] [black, line width=0.2mm] plot [smooth, tension=0] coordinates {(-4,0) (-2.83,2.83)};

\draw [-] [black, line width=0.2mm] plot [smooth, tension=0] coordinates {(0,4) (-2.83,2.83)};

\draw[black,fill=white] (0,4) ellipse (0.5 cm  and 0.5 cm);
\draw[black,fill=white] (4,0) ellipse (0.5 cm  and 0.5 cm);
\draw[black,fill=white] (0,-4) ellipse (0.5 cm  and 0.5 cm);
\draw[black,fill=white] (-4,0) ellipse (0.5 cm  and 0.5 cm);
\draw[black,fill=white] (2.83,2.83) ellipse (0.5 cm  and 0.5 cm);
\draw[black,fill=white] (-2.83,-2.83) ellipse (0.5 cm  and 0.5 cm);
\draw[black,fill=white] (2.83,-2.83) ellipse (0.5 cm  and 0.5 cm);
\draw[black,fill=white] (-2.83,2.83) ellipse (0.5 cm  and 0.5 cm);

\node (1) at (0,4) {{1}};
\node (2) at (2.83,2.83) {{2}};
\node (3) at (4,0) {{3}};
\node (4) at (2.83,-2.83) {{4}};
\node (5) at (0,-4) {{5}};
\node (6) at (-2.83,-2.83) {{6}};
\node (7) at (-4,0) {{7}};
\node (8) at (-2.83,2.83) {{8}};
\end{tikzpicture}
\caption{A Carr-Vempala point with 8 vertices on its cycle. Solid lines are edges with value strictly between $0$ and $1$. Dashed edges represent paths where each edge on the path has $x_e =1$. There can be an arbitrary number of degree-$2$ vertices on the dashed paths.}
\label{fig:CVpoint}
\end{figure}
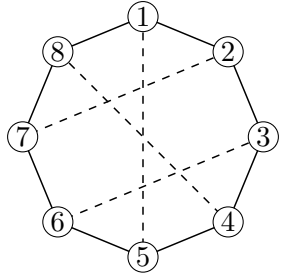

\begin{figure}[h!]
	\centering
	\includegraphics[width=9cm,scale=1.4]{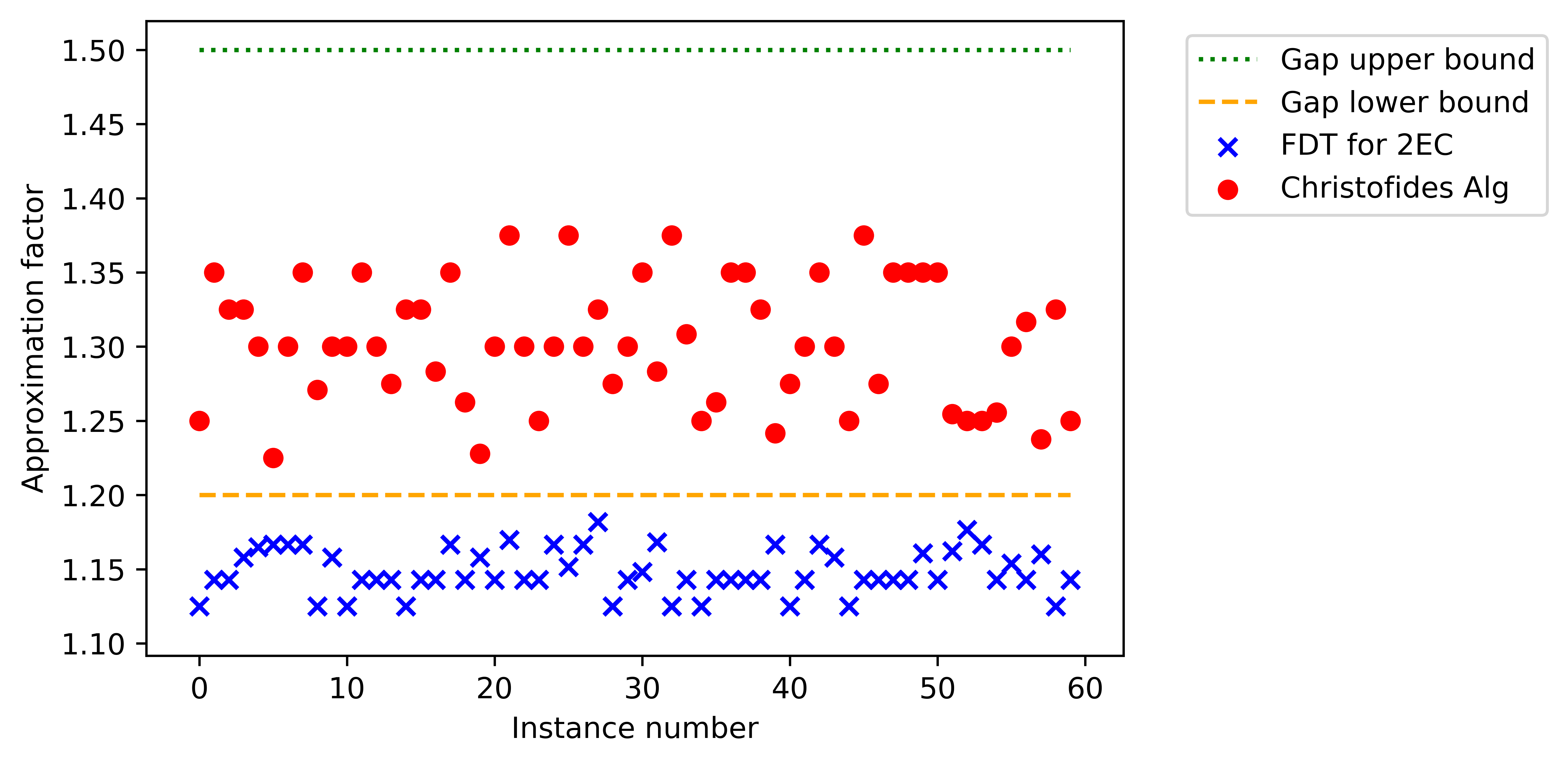}
	\caption{Polyhedral version of Christofides' algorithm vs FDT on all Carr-Vempala points that have 10 vertices on the single cycle formed by fractional edges.}
	\label{fdtvschris}
\end{figure}
\paragraph{FDT for 2EC on Carr-Vempala points.}
We ran FDT for 2EC on 963 fractional extreme points of $\subtour(G)$. We enumerated all (fractional) Carr-Vempala points with $10$ and $12$ vertices. Table \ref{table2EC} shows that again FDT found solutions better than the integrality-gap lower bound for most instances. 
\begin{table}[h!]
	\begin{small}
		\centering
		\begin{tabular}{c c c c c}
			\hline
			& $C\in [1.08,1.11]\;$ & $\;C\in (1.11,1.14]\;$ &
			$\;C\in (1.14,1.17]$ &\; $C\in (1.17,1.2]\;$ \\ \hline
			2EC & $79$ & $201$ & $605$ & $78$ \\ \hline\\
		\end{tabular}	\caption{FDT for $\2ec$ implemented applied to all Carr-Vempala with 10 or 12 vertices. A Carr-Vempala point with $k$ vertices has $\frac{3k}{2}$ edges. Thus, the upper bound provided by Theorem \ref{FDT2EC} is $g(\2ec)^{3k/2}$. The lower bound on $g(\2ec)$ is $\frac{6}{5}$.}
		\label{table2EC}
	\end{small}
\end{table}

\section{Concluding Remarks}

The results in Sections \ref{sec:domTOIP}, \ref{sec:binaryfdt}  and \ref{sec:2EC} hold if the intial integer program has some auxilliary continuous variables\footnote{Here by auxilliary variables we mean a variable that does not participate in the objective function}. That is, when we have 
\begin{equation*}
S(A,b) = \{x\in \mathbb{Z}^n\times \mathbb{R}^p: Ax\geq b\},
\end{equation*}
where the last $p$ variables are continous, $P(A,b) = \{x\in \mathbb{R}^{n+p}: Ax\geq b\}$
and define
\begin{equation*}
g(I) = \max_{c\in \mathbb{R}^n_+}\frac{\min_{x\in S(A,b)} \sum_{i=1}^{n}c_ix_i}{\min_{x\in P(A,b)} \sum_{i=1}^{n}c_ix_i},
\end{equation*}
our main results work. In fact our implementation of the subtour-elimination relaxation is based on an extended formulation with auxiliary variables (see \cite{subtour-extended}). We removed this extension to make the presentation simpler.

Our experiments in Section~\ref{sec:experiment} give a proof of concept. They show that its plausible FDT will have practical benefit as an IP heuristic for problems with appropriate structure. FDT performance will likely improve in a future more high-performance version coded in C or C++, perhaps able to take advantage of the low-level parallelism available on all modern platforms (even laptops).  This will allow tests on larger instances. We leave as future work a more comprehensive set of experiments to determine on which kinds of problems FDT is likely to outperform other general heuristics. Any consistent structure of such instances could lead to proofs of better approximation bounds.

Fractional Decomposition Tree is a tool to experimentally evaluate the known bounds (and conjectured bounds) on the integrality gap of combinatorial optimization IP formulations. We hope that applying this tool to a wide range of problems can be a tool to guide conjectures on the upper bounds of integrality gap or at least narrow down the instances for which a closer look is required for the study of integrality gap.

\cindy{
Here's my original comment on a conclusion section with updated list.
Most papers have a ``discussion and conclusions'' section. It doesn't have to be long and it doesn't have to repeat the paper contributions said before. It can include extensions or thoughts/discussions that don't easily fit into the main body. It can include new insights into the value of work now that the reader understands it better. Here are some thoughts about what could go there:
1)DONE: Bob was quite insistent at one point that we describe how to extend to continuous variables as long as they are not in the objective function. Can we state that in a conclusion and very briefly summarize what would change for this case?
2) Open questions/future research.  Any thoughts on open theory questions?
3) DONE: Experimental questions would include doing a more high-performance version (e.g. in C or C++) to run on larger instances, and a more comprehensive set of experiments to determine on what kind of problems FDT is likely to outperform other general heuristics.
4) TODO: more after I have gone all the way through.
}

\bibliographystyle{alpha}
\bibliography{FDT}

\end{document}